\documentclass[runningheads]{llncs}

\usepackage[utf8]{inputenc}
\usepackage{enumitem}
\usepackage{comment}
\usepackage{amsmath,amssymb,amsfonts}
\usepackage{algorithmic}
\usepackage{graphicx}
\usepackage{textcomp}
\usepackage{bmpsize}
\usepackage{xcolor}
\usepackage{lipsum}
\usepackage{svg}

\title{
Selfish Mining Attacks Exacerbated by\\Elastic Hash Supply}
\author{
 Yoko Shibuya \inst{1} \and
 Go Yamamoto \inst{1} \and
 Fuhito Kojima \inst{1} \and
 Elaine Shi\thanks{This work was performed while the author was consulting for NTT Research during summer 2020.} \inst{2} \and
 \mbox{Shin'ichiro Matsuo} \inst{1} \inst{3} \and
 Aron Laszka \inst{4}}
\authorrunning{
Shibuya et al.}
\institute{
 NTT Research \and
 Cornell University  \and Georgetown University \and
 University of Houston}

\usepackage[textsize=tiny,textwidth=4.2cm]{todonotes}

\newcommand{\rp}{B_\text{attacker}}
\newcommand{\ro}{B_\text{honest}}

\let\xthefootnote\thefootnote
\newcommand{\memo}[2][red]{\def\thefootnote{\color{#1}\xthefootnote}\footnote{\color{#1}#2}}
\newcommand{\void}[1]{}

\newtheorem{prop}{Proposition}

\usepackage[hidelinks, breaklinks]{hyperref}
\usepackage[capitalise]{cleveref}

\begin{document}
\setlength{\marginparwidth}{4.3cm}

\maketitle

\begin{center}
Published in the proceedings of the 25th International Conference on\\Financial Cryptography and Data Security (FC 2021).
\end{center}

\begin{abstract}
\vspace{-1em}
Several attacks have been proposed against Proof-of-Work blockchains, which may increase the attacker's share of mining rewards (e.g., selfish mining, block withholding).
A further impact of such attacks, which has not been considered in prior work, 
is that decreasing the profitability of mining for honest nodes
incentivizes them to stop mining or to leave the attacked chain for a more profitable one.
The departure of honest nodes exacerbates the attack and may further decrease profitability and incentivize more honest nodes to leave. 
%
In this paper, we first present an empirical analysis showing that there is a statistically significant correlation between the profitability of mining and the total hash rate, confirming that miners indeed respond to changing profitability.
Second, we present a theoretical analysis showing that selfish mining under such elastic hash supply leads either to the collapse of a chain, i.e., all honest nodes leaving, or to a stable equilibrium depending on the attacker's initial share.
\vspace{-0.4em}
\end{abstract}

\section{Introduction}\label{sec:Introduction}
\vspace{-0.5em}

When blockchains were first introduced, it was believed that profitable attacks require at least 50\% of the total mining power. However, 
several attacks have been found to go against proof-of-work blockchains, such as selfish mining~\cite{eyal2014majority} and block withholding against mining pools~\cite{eyal2015miner}. A common goal of many such attacks is, at a high level, to increase the attacker's share of the mining rewards by reducing other miners' effective mining power. Prior work found that such attacks may be profitable even if the attacker's original share of the total mining power is less than 50\%.

An important limitation of prior work is that they do not consider how honest miners react to changes in profitability when attacks occur. Most models assume that the total hash supply in a chain is \emph{fixed} and does not respond to changes in the profitability of the chain. In practice, however, most miners are profit-oriented and choose which currency to mine (or to not mine at all) based on their \mbox{profitability}. \void{ If honest miners' profit-oriented behavior is incorporated into a model when attacks reduce honest miner's profitability, we should expect honest miners to leave for other more profitable chains (or to stop mining).  The exit of honest miners increases the attacker's share, which further decreases profitability for honest miners. The negative feedback continues until the honest miner's profitability converges to the original level, or all the honest miners leave the chain.} 

In this paper, we first document real-world evidence of miner's profit-oriented behavior,  using data from three different cryptocurrencies. We found a positive and statistically significant correlation between total hash supply and per-hash mining revenue, i.e., the evidence of  \emph{elastic} hash supply with respect to miners' revenue. We then provide a new analysis of selfish mining that takes into account the elasticity of hash supply. \void{We give an overview
of our analyses and findings below before diving into subsequent, detailed sections. }
In an elegant work by Huberman et al.~\cite{huberman2019economic}, the authors point out that Bitcoin mining is a free-entry, two-sided market. If there is a profit to be made, more miners will enter, which will then trigger the difficulty adjustment algorithm, making mining more difficult, and thus everyone's expected mining revenue decreases. In the equilibrium state, miners break even, i.e., the revenue that they earn from mining is equal to their cost. 
Inspired by this principle, we incorporate a free-entry condition in a model of selfish mining, and thus 
our analysis essentially characterizes the long-term effects of selfish mining on the eco-system in the equilibrium state. 

To understand our analysis, let us first quickly review 
the classical analysis of the selfish mining attack: when a coalition (e.g., a mining pool) mines a new block ${\sf B}^*$ off the current longest chain denoted ${\sf chain}$, it does not release ${\sf B}^*$ immediately but withholds ${\sf B}^*$. Whenever an honest miner mines a block ${\sf B}$ also off ${\sf chain}$, the adversary releases the withheld block ${\sf B}^*$ immediately, and races in transmitting its block ${\sf B}^*$ to other miners. If the adversary has good control of the network (e.g., it controls some relays in the network) and can transmit its block ${\sf B}^*$ ahead of the honest block ${\sf B}$, it can convince other miners to mine off ${\sf B}^*$. In this case, the honest miners' work in mining ${\sf B}$ got erased.  In other words, through selfish mining, an adversary controlling a coalition can erase some fraction of the honest mining power, and therefore the selfish coalition can gain an unfair share of the total rewards. In the worst case, assuming that the coalition can reliably win the race in transmission, then a 1/3 coalition can erase 1/3 of the honest mining power, and thus gain 1/2 of the rewards.

The above classical analysis assumes that the total hash power participating in mining is fixed.
Now let us consider what happens when honest miners may respond to profitability and freely enter and leave the system. 
During a selfish mining attack, because a fraction of the honest mining power is being erased, the erased fraction is essentially not gaining rewards. The immediate effect is that the cost of mining to earn each unit of reward becomes proportionally higher for honest miners; and if the honest miners' profitability plunges below zero, they start to leave the system. As honest miners leave, the impact of the attack on the remaining miners is magnified as a higher fraction of their mining power is now erased, which in turn drives more miners away. At the same time, as honest miners leave, the total mining power decreases. Therefore, the mining difficulty drops, and thus mining becomes cheaper---this second effect somewhat counteract the decreased profitability for honest miners that stems from being the victim of selfish mining.  

When hash power is elastic, what happens in the equilibrium state is driven simultaneously by the above two opposite effects. We show that for a wide range of parameter regimes, the first effect dominates and leads to a ``collapse scenario''---specifically, selfish mining drives costs up for honest miners, and {\it all} honest miners end up leaving the system as a result. In some other parameter regimes, however, because the two effects somewhat counteract each other, the system reaches
a new equilibrium after some but not all honest nodes have left.
In either scenario, the unfairness of selfish mining is significantly exacerbated by the elasticity of hash power.


The rest of the paper is organized as follows. \cref{s:rusing-selifish-miner} explains the intuition of our main result using a simple model of selfish mining.  \cref{sec:Emprical} shows our empirical evidence on the elasticity of hash supply, which motivates our model setting. \cref{sec:Theory} describes our theoretical model of selfish mining under elastic hash supply. \cref{sec:Conclusions} concludes the paper.

\section{\textcolor{black}{Analysis for a ``Rushing'' Selfish Miner}}\label{s:rusing-selifish-miner}

To explain the high-level intuition, we first give a simple analysis for the special case when the selfish miner always succeeds in the network race, i.e., if an honest player sends a block in some round, the adversary can always rush-send an equal-length block, and preempt the honest player's block.

Suppose that in the beginning, the total normalized hashpower is $1$, everyone is mining honestly and the system is in an equilibrium, i.e., everyone breaks even.
Without loss of generality, we may assume that the normalized cost to mine a block is $1$, and the per-block reward is also $1$.
Now, assume that $\alpha \in (0, 1)$ fraction of the hashpower colludes and selfish-mines, while the remaining $1-\alpha$ fraction is honest. We shall assume that the selfish miner's hashpower stays fixed.
As mentioned in Section~\ref{sec:Introduction}, 
for every block mined by the selfish miner, it can erase one block mined by an honest miner.
Such a selfish mining attack effectively erases a subset
of the honestly mined blocks and drives up the cost of mining for honest players. 
As honest players lose money, a subset of them will leave.
Suppose that $\beta < 1-\alpha$ amount of honest hashpower 
leaves.
The remaining honest hashpower is therefore $1-\alpha-\beta$. 

Now, if $\alpha \geq 1-\alpha-\beta$, this means that the selfish miner can erase all remaining hashpower, and thus honest miners earn nothing and they would all leave.
We want to see if there might be any $\beta$ that causes the system to reach a new equilibrium where some positive number of honest miners stay in the system. Therefore, 
henceforth we may assume $\alpha < 1-\alpha-\beta$, i.e, $1-2\alpha-\beta > 0$.
In this case, the total effective mining power
is $1-\alpha-\beta$ taking into account the fact
that $\alpha$ amount of honest hashpower gets erased by the selfish miner.
The blockchain will auto-adjust its difficulty such that the average block-interval is a constant. Thus, the normalized cost of 
mining a block now becomes $1-\alpha-\beta$, and recall that the reward per block is still $1$.
Since $\alpha$ amount of honest hashpower gets erased by the selfish miner, effectively only $\frac{1-2\alpha-\beta}{1-\alpha-\beta}$ fraction of the honestly mined blocks get paid.
Thus, the total profit of the honest players can be represented
by the following expression:
\begin{equation}
\frac{1-2\alpha-\beta}{1-\alpha-\beta} - (1-\alpha-\beta)
= \frac{\beta-(\alpha + \beta)^2}{1-\alpha-\beta}
\label{eqn:honestprofit}
\end{equation}
If the system were to enter a new equilibrium after $\beta$ honest hashpower leaves, then the above expression, i.e., Equation~\ref{eqn:honestprofit} 
should be exactly 0.
We can draw some interesting conclusions from Equation~\ref{eqn:honestprofit}:
\begin{enumerate}
    \item {\bf Collapse regime.} If $\alpha > \frac{1}{4}$, it must be that $\beta - (\alpha + \beta)^2 < 0$ for any choice of $\beta$. In this case, it means that no matter what $\beta$ is, honest players end up losing money and they will eventually all leave, and the systems collapses.
    \item {\bf Some honest hashpower remains.}
    If $\alpha \leq \frac14$, some honest hashpower remains and the system enters a new equilibrium. Consider the special case when $\alpha = \frac14$. In this case, the system reaches a new equilibrium when $\beta = \frac14$ honest hashpower leaves. At this moment, one can verify that the effective total hashpower is $\frac12$, and thus the cost of mining a block becomes $\frac12$. However, half of the remaining honest hashpower gets erased by the $\alpha = \frac14$ selfish-mining, and thus honest players break even.
\end{enumerate}

No matter which case, the unfairness caused by selfish mining becomes exacerbated due to elastic hashpower.

\vspace{-0.25em}
\section{Empirical Findings}\label{sec:Emprical}
\vspace{-0.5em}

This section presents new empirical facts on the elasticity of hash supply with respect to the miners' revenue.
Only few pieces of literature have worked on measuring elasticity of hash supply \cite{noda2020economic}, but our paper distinguishes itself from the prior literature in terms of length and coverage of time-series data. We study the elasticity of hash supply with respect to miners' revenue using data from 3 different currencies (Bitcoin, Ethereum, and Ethereum Classic) from 2015-2020. We use various time-detrending methods from macroeconomics to deal with technological advancements in cryptocurrencies and related time trends in variables. Using time-detrending methods allows us to study longer time periods.

We start by explaining the data and empirical strategy that we use in our study (\cref{subsec:Data_strategy}). We then show the regression results in \cref{subsec:Empirical_result}.

\vspace{-0.25em}
\subsection{Data and Empirical Strategy}\label{subsec:Data_strategy}
\vspace{-0.5em}

\subsubsection{Data}

We downloaded cryptocurrency data from three sources: Bitcoin data from Quandl, Ethereum data from Etherscan, and Ethereum Classic data from crypto-ethereum-classic public library on BigQuery. We use three variables in our regression analysis: daily price, network difficulty, and total hash rate of each cryptocurrency. Different currencies have different lengths of history, and thus we use data from 2017/1/1 to 2020/7/31
for Ethereum and Ethereum Classic, and from 2015/1/1 to 2020/7/31 for Bitcoin. We computed daily per-hash revenue from coinbase using daily price and network difficulty (and data on cryptocurrency halving). We focus on miner's revenue from coinbase and not from transaction fees because transaction fees have been randomly fluctuating over the recent years in these cryptocurrencies.

\subsubsection{Time Detrending Methods}
Technological advancements in cryptocurrency mining over the past 10 years pose a challenge for regression analysis since they add significant time trends to time-series variables. While we would like to observe the change in total hash rate in response to the change in miners' revenue in the short-horizon (e.g. if miners' revenue increased by 1\% today, how much will total hash rate change?), if we measure the correlation between two variables without detrending the time-series data, it will overestimate the correlation because the measured correlation mainly comes from the long-term correlation due to the technological advancements. In order to deal with the trend issue, we apply time-detrending filters that are commonly used in macroeconomics. Time-detrending filters allow us to separate the slow-moving trend component from the shorter-horizon cyclical fluctuations (\cite{hamilton2020time}). By applying time-detrending filters to time-series data of total hash rate and miners' revenue, we can eliminate the long-term trend components in the data such as technology advancements and regulatory developments. We apply three types of time-detrending filters that are commonly used in macroeconomics:  Hodrick-Prescott (HP),  Baxter-King (BK), and  Christiano-Fitzgerald (CF) filters.\footnote{For HP filter, we use $\lambda=10,000$. For BK filter we use (7, 90, 12) for high, low frequencies, and lead-lag length, respectively. For CF filter, we use (7, 90) for high and low-frequency length.}

Figure~\ref{f:Empirical_time_detrending} shows the decomposition of the logarithm of the total hash rate in the Bitcoin network over the past three years, using Hodrick-Prescott filter. The total hash rate of the Bitcoin network has an increasing trend over this time period, and the filter removes out the trend. In the later regression analysis, we use the cycle components of the variables after applying filters. We report regression results based on multiple time-detrending methods.

\begin{figure}[h!]
	\centering
	\includegraphics[width=0.7\textwidth]{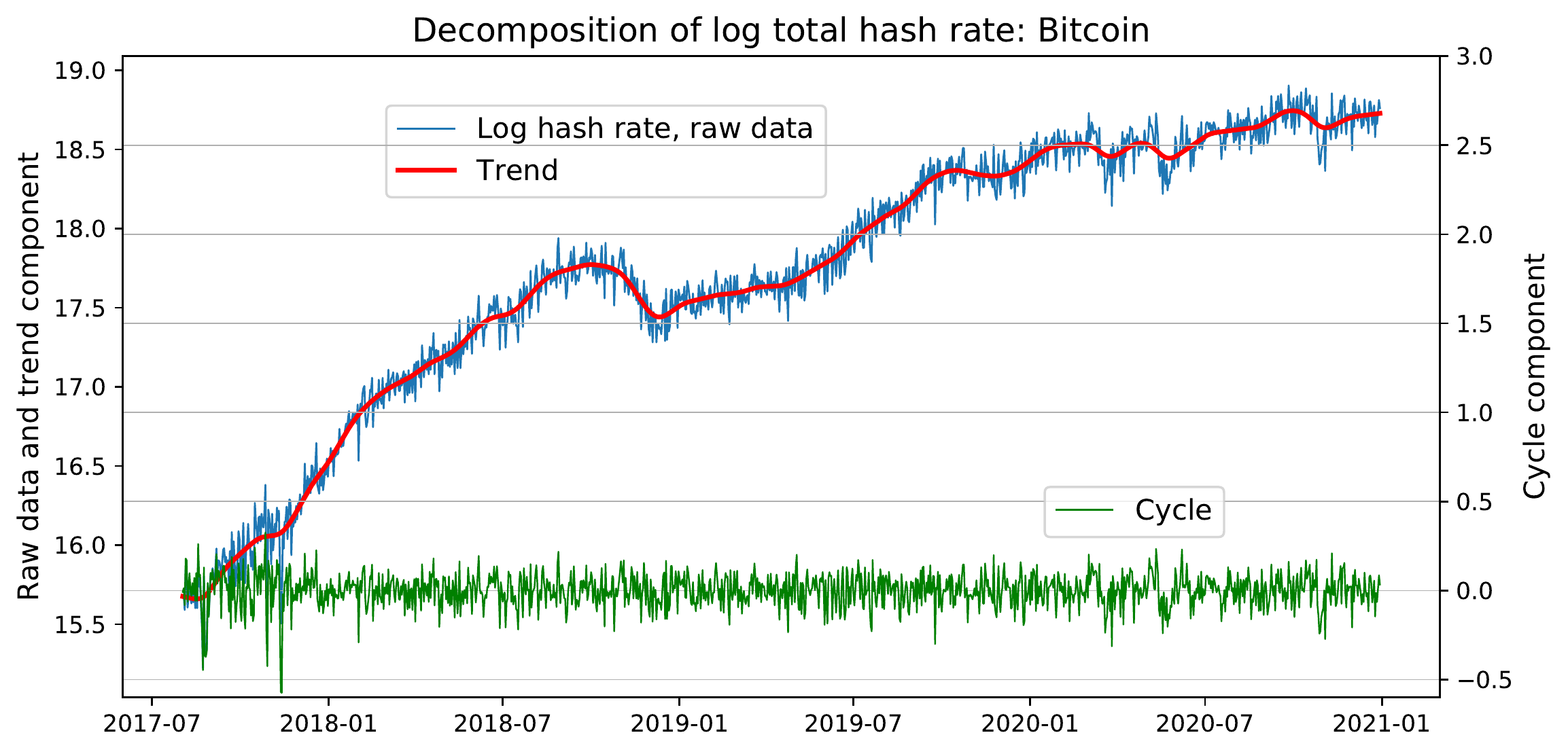}
	\caption{Application of HP filter to raw hash-rate data from Bitcoin.} 
	\label{f:Empirical_time_detrending}
	\vspace{-1em}
\end{figure}
\subsubsection{Regression Equations} To estimate the elasticity of total hash rate with respect to per-hash revenue, we consider the following regression equation:
\begin{equation}
    \Delta \log \text{THR}_{i,t} = \alpha_i \Delta \log \text{MRC}_{i,t} + \epsilon_{i,t},
\end{equation}
where THR stands for total hash rate, MRC stands for miners' per-hash revenue from coinbase. Parameter $i$ is an index representing the cryptocurrency (Bitcoin, Ethereum, or Ethereum Classic), and $t$ is an index for time (day). Variables with $\Delta$ are cycle components of the logged variables.\footnote{For regressions with Ethereum Classic data, we use daily difference in total hash rate as an independent variable. The reasons for this is that total hash rate of Ethereum Classic is volatile at high frequency, and does not exhibit any time trend over the sample period.} We include year-month fixed effect in the regression to take out some year/month fixed events such as regulation changes, which might not be taken out by time-detrending filters.

\vspace{-0.5em}
\subsection{Results}\label{subsec:Empirical_result}
\vspace{-0.5em}

Table \ref{tab:regression_1} summarizes the results of running the above regression for the three cryptocurrencies. The main result of the regression analysis is that with any type of time-detrending filter, in any time period, and for any currency, the coefficients on $\Delta \log{}$MRC are \emph{positive and statistically significant}. In other words, the total hash rate is elastic with respect to the miners' per-hash revenue from coinbase. The magnitude of the coefficient varies across different time-detrending methods and different currencies, but the elasticity ranges from $0.028$ to $0.183$. One percentage change in the miners' per-hash revenue from coinbase causes 0.027 to 0.183 percentage change in the total hash rate.

Regression with a longer sample period for Bitcoin data gives us a more interesting result. Table \ref{tab:regression_2} summarizes the regression results for Bitcoin data with different sample periods. Interestingly, elasticity is higher and more statistically significantly positive in the recent period (2018--2020) compared to the beginning of the sample period (2015--2017). This shows the possibility that the hash rate becomes more responsive to the miners' revenue as a currency grows.
\begin{table}[h!]
\centering
\caption{Regression results for three currencies in sample period 2017/1/1--2020/7/31}
\label{tab:regression_1}
\fontsize{3.5mm}{5mm}
\selectfont
\resizebox{\textwidth}{!}{%
\begin{tabular}{cccccccccc}
\hline
    & \multicolumn{3}{c}{\textbf{Bitcoin}} & \multicolumn{3}{c}{\textbf{Ethereum}} & \multicolumn{3}{c}{\textbf{Ethereum Classic}} \\
    & HP      & BK      & CF      & HP       & BK      & CF      & HP         & BK         & CF         \\ \cline{2-10} 
$\Delta\log{}$MRC &
  \begin{tabular}[c]{@{}c@{}}$0.175^{***}$\\ (5.53)\end{tabular} &
  \begin{tabular}[c]{@{}c@{}}$0.183^{***}$\\ (8.83)\end{tabular} &
  \begin{tabular}[c]{@{}c@{}}$0.181^{***}$\\ (1.30)\end{tabular} &
  \begin{tabular}[c]{@{}c@{}}$0.028^{***}$\\ (3.69)\end{tabular} &
  \begin{tabular}[c]{@{}c@{}}$0.033^{***}$\\ (5.08)\end{tabular} &
  \begin{tabular}[c]{@{}c@{}}$0.079^{***}$\\ (12.54)\end{tabular} &
  \begin{tabular}[c]{@{}c@{}}$0.041^{***}$\\ (3.20)\end{tabular} &
  \begin{tabular}[c]{@{}c@{}}$0.048^{***}$\\ (3.12)\end{tabular} &
  \begin{tabular}[c]{@{}c@{}}$0.027^{***}$\\ (2.57)\end{tabular} \\
 No. of obs. & 1308    & 1296    & 1308    & 1308     & 1296    & 1308    & 1308       & 1296       & 1308  \\ \hline     
\multicolumn{10}{l}{\textsuperscript{***}$p<0.01$, 
  \textsuperscript{**}$p<0.05$, 
  \textsuperscript{*}$p<0.1$, t-values in parentheses.}
\end{tabular}
}
\end{table}
\begin{table}[h!]
\centering
\caption{Regression results for Bitcoin data with three different sample periods}
\label{tab:regression_2}
\fontsize{3.5mm}{5mm}
\selectfont
\resizebox{\textwidth}{!}{%
\begin{tabular}{cccccccccc}
\hline
    & \multicolumn{3}{c}{\textbf{2015/1 - 2017/12}} & \multicolumn{3}{c}{\textbf{2018/1 - 2020/7}} & \multicolumn{3}{c}{\textbf{2015/1 - 2020/7}} \\
    & HP      & BK      & CF      & HP       & BK      & CF      & HP         & BK         & CF         \\ \cline{2-10} 
$\Delta\log{}$MRC &
  \begin{tabular}[c]{@{}c@{}}$0.082^{*}$\\ (2.02)\end{tabular} &
  \begin{tabular}[c]{@{}c@{}}$0.108^{***}$\\ (3.83)\end{tabular} &
  \begin{tabular}[c]{@{}c@{}}$0.078^{***}$\\ (3.62)\end{tabular} &
  \begin{tabular}[c]{@{}c@{}}$0.163^{***}$\\ (4.85)\end{tabular} &
  \begin{tabular}[c]{@{}c@{}}$0.152^{***}$\\ (6.88)\end{tabular} &
  \begin{tabular}[c]{@{}c@{}}$0.194^{***}$\\ (11.76)\end{tabular} &
  \begin{tabular}[c]{@{}c@{}}$0.126^{***}$\\ (4.80)\end{tabular} &
  \begin{tabular}[c]{@{}c@{}}$0.132^{***}$\\ (7.38)\end{tabular} &
  \begin{tabular}[c]{@{}c@{}}$0.143^{***}$\\ (10.65)\end{tabular} \\
No. of obs. & 1096    & 1084    & 1096    & 943     & 931    & 943    & 2039       & 2015       & 2039  \\ \hline    
\multicolumn{10}{l}{\textsuperscript{***}$p<0.01$, 
  \textsuperscript{**}$p<0.05$, 
  \textsuperscript{*}$p<0.1$, t-values in parentheses.}
\end{tabular}
}
\end{table}

\vspace{-0.5em}
\section{Model with Elastic Hash Supply}\label{sec:Theory} 
\vspace{-0.5em}

This section introduces our model of selfish mining with elastic hash supply. \void{Under elastic hash supply, total hash supply is determined endogenously by \emph{free entry} condition (see an elegant work by \cite{huberman2019economic} that pointed out that Bitcoin mining is a free-entry two-sided market). Free entry means that miners keep entering the system as long as there is a positive profitability, while they keep leaving the system as long as the profit is negative; thus, the profitability in the system always reaches  zero in equilibrium. } We first explain our baseline model without selfish mining and illustrate how total hash rate is determined endogenously in an equilibrium by free-entry condition. We then analyze the model with selfish mining, building on the seminal work by Eyal and Sirer \cite{eyal2014majority}. Lastly, we discuss the stability of equilibria. 

In our model of selfish mining with elastic hash supply, the equilibrium state is determined by the two opposing effects. An attack increases the cost of mining for honest miners and thus makes honest miners leave. At the same time, when some honest miners leave, the total mining power decreases and so does the cost of mining for honest miners. 
\footnote{\textcolor{black}{In practice, as miners leave and the total mining power decreases, the price of the cryptoccurrency may drop, thereby decreasing the revenue of the remaining honest miners. Similar to the attacker's increasing share due to miners leaving, this effect exacerbates the impact of the attack, and in this sense, magnifies the phenomenon that we identified in this paper. We leave the modeling and formal analysis of this effect to future work.}}
Which effect dominates depends on the attacker's initial share of mining power. We derive a threshold for the attacker's initial share such that (a) if the attacker's share is below the threshold, the system has a stable equilibrium with a positive hash supply by honest miners; and (b) if the attacker's share is above the threshold, all honest miners leave and the system collapses. In either case, some or all honest miners leave the system, and thus the effect of selfish mining is significantly exacerbated under elastic hash supply. 

\textcolor{black}{
\subsubsection{Notations}
The following notations and basic assumptions are employed:
\begin{itemize}
    \item $B =$ expected reward for a new block, including both the coinbase and the transaction fee. For example, as of December 2020, $B$ is 6.25 BTC $\approx$ 169,441~USD coinbase plus transaction fees for Bitcoin.
    \item $C =$ expected cost of mining per unit of hash-rate until some miner finds a new block. This includes electricity costs, depreciation, and other operational costs.  We assume that a miner's cost is proportional to its hash rate, and the cost per unit of hash-rate is the same for all miners. 
    For example, $C$ is the cost per unit of hash-rate for 10 minutes for Bitcoin.  On the online marketplace NiceHash (\texttt{nicehash.com}), as of December 2020, the lowest-price offer for 1 PH/s of mining power\footnote{We use PH/s and EH/s to denote peta-hash per second and exa-hash per second.} for 24 hours is 0.0069 BTC. From this, we can estimate $C$ as $0.0069~\text{BTC / (PH/s) / 24 hours} \cdot 10~\text{minutes} \approx 1.31~\text{USD / (PH/s)}$. 
    \item $H =$ honest miners' hash rate in total.
    \item $M =$ attacking pool's hash rate.
\end{itemize}
}

\subsubsection{Baseline Model Without Selfish Mining}\label{sub:Baseline_model}
We consider a system with a group of honest miners (with mining power $H$) and an attacking pool (with mining power $M$). 
\void{We let $B$ and $C$ denote block rewards and per-hash mining cost, respectively.} We assume \emph{elastic hash supply} in the system: the equilibrium mining power of honest miners ($H^*$) is determined such that honest miners make zero profit with mining power $H^*$. The attacking pool's mining power ($M$), block rewards ($B$) and cost ($C$) are assumed to be fixed and to satisfy $M<B/C$.
\void{
\memo{
List of notations
\begin{itemize}
    \item $B$[USD/block]: Mining reward per block including both the coinbase and the transaction fees.
    \item $C$[USD/(EH/s)/block]: Mining cost per block per hash rate, including electricity cost, depreciation and so on.
    \item $H$[EH/s]: The honest miner's hash rate in total.
    \item $M$[EH/s]: The attacking pool's hash rate.
    \item $B_{\text{honest}}$[USD/block]: The honest miners' total expected mining reward per block discovery including hidden ones.
    \item $B_{\text{attacker}}$[USD/block]: The honest attacker's expected mining reward per block discovery including hidden ones.
    \item ${\cal{U}}^S(H)$[USD/(EH/s)/block]: Per hash profit per block discovery \underline{excluding} hidden ones after difficulty adjustment.
\end{itemize}
}
}

Without selfish mining attack, the honest miners' profit per unit hash rate
is
%
\begin{equation}
    \mathcal{U}^N(H) = B \frac{1}{H + M} - C.
\end{equation}\label{eq:profit_honest_without_attack}
In an equilibrium, the elastic hash supply assumption implies $\mathcal{U}^N(H^*)=0$. We can solve for $H^*$:
\begin{equation}
H^* = \frac{B}{C} - M > 0
\end{equation}

\subsubsection{Model With Selfish Mining}\label{sub:Model_with_selfish_mining}

Now, we assume that the attacking pool performs selfish mining as defined by \cite{eyal2014majority}.
\footnote{It is well-known that selfish mining attack proposed by \cite{eyal2014majority} is not the optimal attacker strategy, and thus please note that the actual equilibrium may be different.}
We can calculate the expected mining reward per block discovery, including the hidden block discoveries, as
%
$$
\rp = B \frac{\left(-2\,\alpha^4+5\,\alpha^3-4\,\alpha^2+\alpha\right)\,\gamma+4\,\alpha^4-9\,\alpha^3+4\,\alpha^2}{2\,\alpha^3-4\,\alpha^2+1}
$$
for the attacking pool, and 
$$
\ro = B \frac{\left(2\,\alpha^4-5\,\alpha^3+4\,\alpha^2-\alpha\right)\,\gamma-4\,\alpha^4+10\,\alpha^3-6\,\alpha^2-\alpha+1}{2\,\alpha^3-4\,\alpha^2+1}
$$
for the honest miners, where we denote by $\alpha = \frac{M}{H+M}$ the fraction of the attacking pool's mining power out of the total mining power, and by $\gamma$ the ratio of honest miners that
choose to mine on the attacking pool's block. 
The total effective mining power in the system under attack is $(\ro+ \rp) (H+M) /B$. \footnote{These calculations should coincide $\rp = B \cdot r_\text{pool}$ and $\ro = B \cdot r_\text{others}$, where $r_\text{pool}$ and $r_\text{others}$ are from Equations (6) and (7) in \cite{eyal2014majority}.}
The honest miners' effective mining power is given by $\ro (H+M)/B = \frac{\ro}{(1-\alpha) B} H$ and the attacking pool's effective mining power is $\rp (H+M) /B = \frac{\rp}{\alpha B} M$.

Then, the honest miners' per hash-rate profit under selfish mining attack is 
\begin{equation}
    \mathcal{U}^S(H) = B \frac{\ro}{(1-\alpha)B} \frac{B}{ (\rp+\ro)(H + M)} - C.
\end{equation}
In an equilibrium, honest miners' hash supply is again derived from $\mathcal{U}^S(H^*)=0$:
\begin{align}
\mathcal{U}^S(H^*) 
 = B\frac{1}{M} \left\{\frac{\alpha^* \cdot \ro(\alpha^*)}{(1-\alpha^*) (\rp(\alpha^*)+\ro(\alpha^*))} - \kappa\right\} 
 = 0 \label{eq: equilibrium}
\end{align}
for $\alpha^* = \frac{M}{H^* + M}$ and $\kappa = M \cdot \frac{C}{B}$. 

A natural question is whether the above equilibrium condition has a solution $H^*>M$.  If not, then the system cannot find an equilibrium where honest miners stay in the system under selfish mining attack.   This simple theorem answers that the attacker's hash rate must be bounded to avoid collapsing the system.

\begin{theorem}\label{theorem:threshold_initial_share}
For any given $\gamma$, there exists $M_\text{max}$ such that a solution $H^*$ of ${\cal{U}}^S(H^*) = 0$ with $H^* > M (> 0)$ exists if and only if $ M \le M_\text{max}$.
\end{theorem}

\begin{proof}
\textcolor{black}{
$\rp(\alpha)$ and $\ro(\alpha)$ have the following properties:
(A) $\ro(\alpha)$ and $\rp(\alpha)$ are continuous for $0\le \alpha\le 1/2$,
(B)  $\ro(1/2) = 0$, 
and (C) $\ro(\alpha) + \rp(\alpha) > 0$ for all $0\le\alpha\le 1/2$.
Let us define a function 
$
f(\alpha) =  \frac{\alpha \cdot \ro(\alpha)}{(1-\alpha) (\rp(\alpha)+\ro(\alpha))}
$
. 
First, $f(\alpha)$ is continuous for $0\le \alpha \le 1/2$ because of property (A) and property (C).
Since $f$ is continuous, there exists $\alpha_\text{max}\in [0, 1/2]$ that achieves the maximum of $f(\alpha)$ for $0\le \alpha \le 1/2$.  Let $M_\text{max} = \frac{B}{C}f(\alpha_\text{max})$.
If $M_\text{max} <M$, then  ${\cal{U}}^S(H) <0 $ for all $H$ such that $H>M$, so solution $H^*$ does not exist.
To complete the proof it suffices to find a solution $H^* >M$ of ${\cal{U}}^S(H) = 0$ for constant $M$ that satisfies $0< M \le M_\text{max}$.  There exists some $\alpha^* \in (0, 1/2)$ such that $f(\alpha^*) = \frac{C}{B}M$ because $f$ is continuous, $f(0) =f(1/2) = 0$ by property (B), and $0 < \frac{C}{B}M \le \frac{C}{B}M_\text{max}=f(\alpha_\text{max})$. We find $H^*$ by solving $\alpha^* = \frac{M}{H^* +M}$, and $H^* > M$ because $\alpha^* <1/2$.
}
\qed
\end{proof}
We can find $\alpha_\text{max}$ by solving $f'(\alpha) =0 $:  
\begin{align}\label{argmax}
\gamma &=\frac{4\,\alpha^6-16\,\alpha^5+26\,\alpha^3-16\,\alpha^2+1}{2\,\alpha^6-8\,\alpha^5-\,\alpha^4+14\,\alpha^3-10\,\alpha^2+2\,\alpha} .
\end{align}
\void{
For a given $\gamma$, the function  $f(\alpha)$ attains the only maximum on the solution of Equation (\ref{argmax}). Define $\alpha_{max}(\gamma)$ as the value of $\alpha$ that attains the only maximum. 
Since $f(0) = 0$ and $f(1/2)=0$ for any $0 \le \gamma \le 1$,
 there are solutions to Equation (\ref{argmax}) only if  $f(\alpha_{max}) > \kappa=M \frac{C}{B}$, which determines $M_\text{max}(\gamma)$.
}

\void{Theorem \ref{theorem:threshold_initial_share} states that, for a given value of $\gamma$, there is a threshold size of the attacking pool ($M_\text{max}$) such that if the initial size is above the threshold, all the honest miners leave and the system collapses.} 
With elastic hash supply, selfish mining attacks reduce the profitability of honest miners, making honest miners leave the system,  which in turn increases the attacker’s share, further decreasing profitability for honest miners. If the attacking pool's share is large enough, the negative propagation effect forces all  honest miners to leave the system. For example,
when $\gamma = 1$,%
\void{solving equation (\ref{argmax})  with  $\gamma = 1$ will give us $\alpha_\text{max}$ and there exists a solution $H^*> M$ if and only if
\begin{equation}
    M \leq M_\text{max}(\gamma) = f(\alpha_\text{max}) \frac{B}{ C} .
\end{equation}
}
we find that\void{$\alpha_\text{max}$ is approximately $0.378$ and} $f(\alpha_\text{max})$ is approximately $0.292$. 
\void{
With $\gamma = 1$ (when all honest miners choose to mine on the attacking pool’s block), i
}
Since $H^*+M=\frac{B}{C}$ in the equilibrium without selfish mining attacks, this implies that if
the attacking pool's share is larger than 29.2\%, then the attack makes all the honest miners eventually leave the system. When $0\le \gamma \le 1$, $f(\alpha_{max})$ is decreasing in $\gamma$, ranging from $0.3475$ at $\gamma=0$ to $0.2919$ at $\gamma=1$.\void{ The result implies that if attacking pool's share is larger than 29.2\% of total mining power, the system can potentially collapse by selfish mining attack.}

When the system does not collapse, we can find a stable equilibrium from the honest miners' response.
It is straightforward to check that \cref{argmax} has only one solution in $0\le \alpha \le 1/2$.\footnote{We omit the details due to the restriction of space.}  This implies that we have only two equilibria $H^*_1$ and $H^*_2$ when $M< M_\text{max}$.  We assume $H^*_1 < H^*_2$ without loss of generality.  Since $f'(\frac{M}{H^*_1+M}) < 0$ and $f'(\frac{M}{H^*_2+M}) > 0$, we obtain the following proposition.

\begin{prop}
For any given $\gamma$ and $M < M_{max}$, there are two equilibria, $H^*_1$ and $H^*_2$ ($H^*_2>H^*_1$), where  $H^*_2$ is stable and $H^*_1$ is unstable.
\end{prop}

Figure~\ref{f:Theory_honest_revenue_cost} illustrates the honest miners' per-hash revenue and cost, given parameters $B$, $C$, $\gamma$, and $M$.\footnote{\textcolor{black}{We set $B=169,441  \text{ USD}$, $C=1.31 \text{ USD/(PH/s)}$, $\gamma=1$, and $M=0.25 \frac{B}{C}$ for Figure~\ref{f:Theory_honest_revenue_cost}. We computed values of $B$ based on Bitcoin price and block-reward and $C$ based on the most competitive offer from nicehash.com on December 29th, 2020.}} Under the free entry condition, the equilibria correspond to points $H^*_1$ and $H^*_2$ where the revenue curve intersects the cost, i.e., points with zero profit. In this case, equilibrium $H^*_2$ is stable, while $H^*_1$ is not. When honest miners' mining power increases (decreases) by any small amount  $\epsilon>0$ from point $H^*_1$, positive (negative) profit will be generated and more honest miners will enter (leave) the system, ending up reaching equilibrium $H^*_2$ (or an equilibrium $H=0$).\footnote{While $H=0$ is an equilibrium, we do not consider cases where $H<M$ in our analysis since it is well known that such cases are unsustainable.} On the other hand, when mining power increases (decreases) from point $H^*_2$, negative (positive) profit will be generated and honest miners leave (enter) the system. Therefore, equilibrium $H^*_2$ is the only stable equilibrium.  
\begin{figure}[h!]
	\centering
	\includegraphics[width=0.7\textwidth]{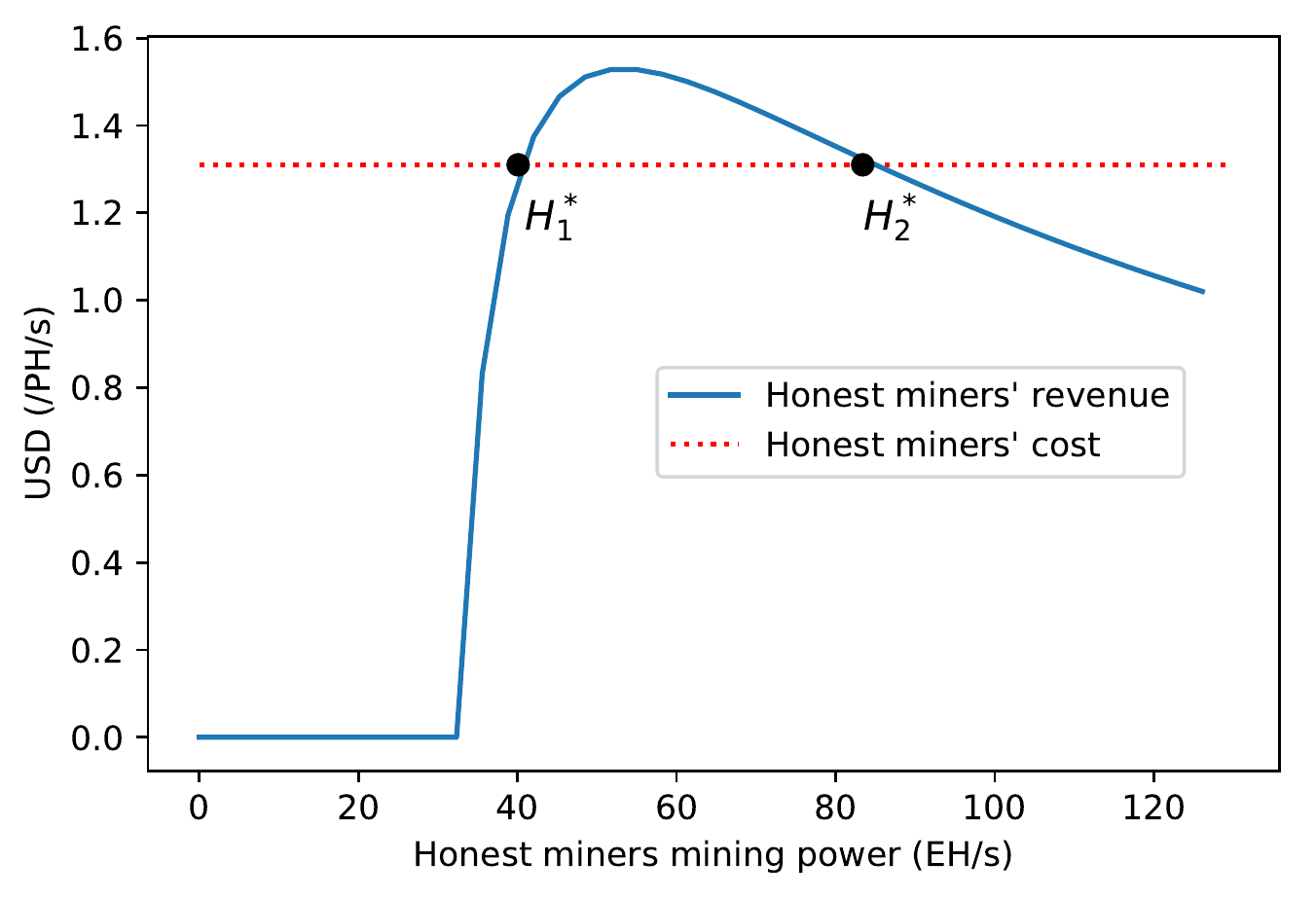}
	\vspace{-5mm}
	\caption{\textcolor{black}{Honest miners' per hash-rate revenue and cost (PH/s: peta-hash per second, EH/s: exa-hash per second)}} 
	\label{f:Theory_honest_revenue_cost}
	\vspace{-2em}
\end{figure}

\section{Conclusions}\label{sec:Conclusions}
The majority of selfish mining literature assumes \emph{fixed} total hash power. Our results show that  \emph{elastic} hash supply can significantly exacerbate the impact of selfish mining.  We (i) showed empirically that hash supply is elastic with respect to the miners' per-hash  revenue and (ii) theoretically derived a threshold such that if the attacker's initial share of the total mining power is above the threshold, all the honest miners will leave and the chain collapses. 

Limitations of our theoretical analysis lead us to future work. First, whether the equilibrium can be reached depends on the starting state. For example, if $H = 0$, then it will not be profitable for any individual honest miner to start mining if its hash power is less than $M$ (regardless of the relation between $M$ and $\frac{B}{C}$). Second, our analysis ignored transient effects, which may prevent reaching particular equilibria. For example, if difficulty adjustments are delayed, an attacker with $M \leq f(\alpha_\text{max})\frac{B}{C}$ might be able to ``chase away'' honest miners before difficulty adjusts to incentivize them to stay. Our future work includes extending our model to dynamic analysis considering initial state and transient effects.

\bibliographystyle{splncs04}
\bibliography{main.bib}
\end{document}